\documentclass[letterpaper, 10 pt, conference]{ieeeconf}  % Comment this line out if you need a4paper
\usepackage{comment}
\usepackage{cite} 
\let\labelindent\relax
\usepackage{enumitem}

\usepackage{amssymb, mathbbol, mathtools, amsthm} 
\usepackage{graphicx}
\usepackage{csquotes} % for quotation marks 
\usepackage{hyperref} % for hyperlinks
\usepackage{xcolor} % for colorful font
\usepackage{textcomp}
\usepackage[export]{adjustbox}
\usepackage{subcaption}
\usepackage{setspace}
\usepackage{float}
\usepackage{array}
\usepackage{tikz}
\usepackage{graphics} % for pdf, bitmapped graphics files
\usepackage{epsfig} % for postscript graphics files
\usepackage{caption}
\usepackage{mathptmx} 
\usepackage{times} % assumes new font selection 
\IEEEoverridecommandlockouts                              % This command is only needed if 
\newtheorem{theorem}{Theorem}
\newtheorem{lemma}{Lemma}
\theoremstyle{definition}
\newtheorem{definition}{Definition}
\newtheorem{example}{Example}
\theoremstyle{corollary}
\newtheorem{corollary}{Corollary}
\title{\LARGE \bf Are the flows of complex-valued Laplacians and their pseudoinverses related?}
\author{Aditi Saxena$^{1}$ \hspace{30 mm} Twinkle Tripathy$^{2}$\hspace{30 mm} Rajasekhar Anguluri$^{3}$
\thanks{$^{1}$Aditi Saxena is a research scholar with the Department of Electrical Engineering, Indian Institute of Technology Kanpur, India. (email: {\tt\small saditi23@iitk.ac.in)}}%
\thanks{$^{2}$Twinkle Tripathy is with the Department of Electrical Engineering, Indian Institute of Technology Kanpur, India. (email: 
{\tt\small ttripathy@iitk.ac.in).}}%
\thanks{$^{3}$Rajasekhar Anguluri is with the Department of Computer Science and Electrical Engineering,  University of Maryland, Baltimore County, MD 85281, USA. (e-mail: {\tt\small rajangul@umbc.edu}).}
        }
\begin{document}
\maketitle
\thispagestyle{empty}
\pagestyle{empty}
\begin{abstract}
Laplacian flows model the rate of change of each node's state as being proportional to the difference between its value and that of its neighbors. Typically, these flows capture diffusion or synchronization dynamics and are well-studied. Expanding on these classical flows, we introduce a \emph{pseudoinverse Laplacian flow system}, substituting the Laplacian with its pseudoinverse within complex-valued networks. Interestingly, for undirected graphs and unsigned weight-balanced digraphs, Laplacian and the pseudoinverse Laplacian flows exhibit an interdependence in terms of consensus. To show this relation, we first present the conditions for achieving consensus in the \emph{pseudoinverse Laplacian flow system} using the property of real eventually exponentially positivity. Thereafter, we show that the pseudoinverse Laplacian flow system converges to consensus if and only if the Laplacian flow system achieves consensus in the above-mentioned networks. However, these are only the sufficient conditions for digraphs. Further, we illustrate the efficacy of the proposed approach through examples, focusing primarily on power networks. 
\end{abstract}
\section{INTRODUCTION}
Network science examines interconnected systems, where nodes represent entities and edges represent interactions. A key focus is on complex networks, which find applications in areas such as social, biological, and technological domains. In these applications, one is often interested in understanding the behavior of nodes (or agents) over time. This behavior is captured by Laplacian flows, where the evolution of node’s state is influenced by its interactions with neighboring nodes, providing insights into insights into dynamics like diffusion, synchronization, and influence propagation \cite{gates2019nature,dorfler2012synchronization,olfati2007consensus}.%\rajmargin{could you cite some high-level journals, that highlight these claims we are stating?--done} 

On the other hand, consensus algorithms focus on enabling a group of agents to reach agreement on a common value or state. Thus, these algorithms can be viewed as a specific case of Laplacian flows, where the Laplacian matrix governs the state evolution attains consensus. Over the last two decades, many important results for consensus algorithms were established by exploring algebraic, spectral, and combinatorial properties of the Laplacian matrices. For studies in this direction and beyond, see \cite{degroot1974reaching,bullo2018lectures,hu2024best,zhang2020consensus,cetinay2016topological,romeres2013novel,de2018power}.
Despite the extensive research on Laplacian flows and their key role in consensus studies, most existing work focuses on ``real-valued" Laplacian matrices. While effective in many applications, such assumption can be restrictive for emergin fields like distribution power networks, quantum dynamics, electrodynamics, and machine learning \cite{muranova2020electrical,song2019extension,liu2015hermitian,bottcher2024complex}. As highlighted in \cite{bottcher2024complex}, the Laplacian matrix in these contexts is complex-valued. We refer to such systems as complex-valued networks to distinguish them from complex networks. The dedicated session on complex-valued networks at SIAM’s 2024 Mathematics of Data Science conference underscores the growing importance of this emerging field.

 % Hence, significant literature studies about the algorithms for achieving consensus in networks with real-valued edge-weights. However, networks with complex edge-weights (\emph{complex-valued networks} hereafter) arise in certain applications like power networks, quantum information, social networks, electrodynamics and machine learning \cite{muranova2020electrical,song2019extension,liu2015hermitian,bottcher2024complex}.
Distributed control algorithms for complex networks often employ Laplacian dynamics to study consensus in social networks \cite{olfati2007consensus}. 
Some works exploring Laplacian dynamics in complex-valued networks include \cite{bottcher2024complex,dong2014complex,DONG20161,tian2024spreading,saxena2024realeventualexponentialpositivity}. In \cite{dong2014complex}, the authors discuss modulus consensus using the concept of structural balance under the assumption of Hermitian adjacency matrices. 
%In \cite{DONG20161}, the authors provide necessary and sufficient conditions for the Laplacian matrix to be singular, based on the topology of a graph in graphs. 
Further, a spectral clustering algorithm in signed networks is proposed in \cite{tian2024spreading} employing structural balance and anti-balance notions for Hermitian weight matrices. Few works employing eventually exponentially positivity property in the analysis of real-valued Laplacian flows are \cite{fontan2021properties,fontan2023pseudoinverses,altafini2019investigating}. %The authors use the former property in conventional (real-valued) networks.%\rajmargin{This paragraph in blue should be re-written exclusively for complex-valued networks.--ok}

In contrast to the existing literature, we study the relationship between complex-valued Laplacians and their pseudoinverses in various types of network, such as signed and directed networks. We compute the Moore-Penrose pseudoinverse of the Laplacian matrix, corresponding to the network, to obtain the pseudoinverse Laplacian flows. We consider unsigned networks which find applications in power networks and non-uniform Kuramoto oscillators. To the best of our knowledge, we are the first to introduce and study pseudoinverse flow systems. Below, we summarize our key contributions:
%we propose an alternative agreement protocol using the pseudoinverses of Laplacian matrices.  \cite{bottcher2024complex,dorfler2012synchronization}. 
% In \cite{dorfler2012synchronization}, the authors state that in power networks with passive admittances, the self-admittance of a generator has non-negative real and imaginary parts.
%The main contributions of the paper are listed below: 
\begin{enumerate}
\item  Using the real eventually exponentially positivity property (rEEP) introduced in  \cite{saxena2024realeventualexponentialpositivity}, we establish the necessary and sufficient conditions for achieving consensus in both complex-valued Laplacian and pseudoinverse Laplacian flows. We focus on cooperative and antagonistic (i.e., symmetric) complex-valued networks and weight-balanced digraphs. 
\item We establish that the property of rEEP is equivalent to the marginal stability or semi-convergence of the Laplacian, as well as the pseudoinverse Laplacian flows. We use complex Perron-Frobenius theory to establish the real eventually exponential positivity of the matrices. %\rajmargin{What is the use of the blue highlighted part in paragraph? Please highlight your contributions, not your motivations!--ok}\
%\rajmargin{why is that computing pseduo inverse a contribution?--I mentioned it to clarify how we calculate pseudoinverse flows corresponding to a graph}
\end{enumerate}
 Finally, we validate our mathematical findings on synthetic and IEEE benchmark power networks. In addition, throughout the paper we demonstrate our theoretical results using several toy network systems. 
%Although the diagonal elements of the matrix exponential of the Laplacian matrix can be used to evaluate the importance of node based on the number of closed walks in a network, we identify this as a potential direction for future research.
%\emph{Related}
\section{PRELIMINARIES}\label{sec:preliminaries}
\subsection{Matrix Algebra}
We denote the space of $n\times n$ complex-valued matrices by $\mathbb{C}^{n \times n}$. For any $m$ (scalar or matrix), ${\Re}(m)$ and ${\Im}(m)$ indicates the real and imaginary part of $m$ respectively. A matrix $M:=[m_{ij}] \in  \mathbb{C}^{n \times n}$ is non-negative (positive) if ${\Re}(m_{ij})\geq 0$ and ${\Im}(m_{ij})\geq 0$ (respectively ${\Re}(m_{ij})> 0 \text{ and }{\Im}(m_{ij}> 0)$) $\forall i,j \in \left \{ 1,2,...,n \right \}$. A matrix is real non-negative (positive), that is, $\Re(M)\geq0~ (\Re(M)>0 $) if $\Re{(m_{ij})}\geq 0$. $M$ is real eventually non-negative (positive) if $\Re(M^{k})\geq 0~ (\Re(M^{k})>0)$ holds for some $ k\geq k_{0}$. $M^{H}$ is the conjugate transpose or adjoint of $M$. $M$ is self-adjoint if $M=M^H$ holds. The spectrum of $M$ is denoted by $\operatorname{spec}(M)=\left \{ \lambda_{1},\lambda_{2},...,\lambda_{n} \right \}$, where $\lambda_{i}$ is the $i^{th}$ eigenvalue of $M$. The symmetric part of $M$ is denoted as $M_s$. Any symmetric matrix $M$ is positive semi-definite $(psd)$ if all eigenvalues are in the closed right-half plane (RHP). The spectral radius of a matrix $M$ is the smallest non-negative real number $\rho(M)\geq |\lambda_{j}(M)|$ for all $\lambda_{j} \in \operatorname{spec}(M)$. Corank of $M$ is defined as the dimension of its null space \cite{horn2012}.
\subsection{Graph Theory}
\looseness=-1
Consider the weighted directed graph $\mathcal{G}(A)=(\mathcal{V},\mathcal{E},A)$, where $\mathcal{V}$ is the set of vertices, $\mathcal{E}$ is the set of edges, and $A=[a_{ij}] \in  \mathbb{C}^{n \times n}$ is the adjacency matrix. If $a_{ij}\geq 0$, then $\mathcal{G}(A)$ corresponds to an unsigned network. Otherwise, it is a signed network. In mathematical sociology, the positive edges in a network represent friendship, while the negative edges indicate enmity between individuals. 
%When the adjacency matrix consists of non-negative edge weights then $A$ corresponds to an unsigned network. 
%Degree matrices are diagonal matrices with degrees of the nodes as their diagonal elements. In case of digraphs, out-degree of a node is the sum of edge weights of all the outgoing edges. 
Degree matrices are diagonal matrices where each diagonal entry represents the degree of a corresponding node. For digraphs, the out-degree of a node is the sum of the weights of all its outgoing edges.
Out-degree matrix is given by $D_{out} = diag(A \mathbb{1}_n)$. The Laplacian matrix $L$ is defined by $L=D_{out}-A$. The Moore-Penrose pseudoinverse (hereafter Laplacian pseudoinverse) of $L$ is denoted by $L^{\dagger}$. 
A digraph (or undirected graph) $\mathcal{G}(A)$ is said to be strongly connected (or connected),  if any node can be reached from any other node. A digraph is weight-balanced if each node’s in-degree equals its out-degree.
$\mathcal{G}(A)$ is said to be unsigned when all its edge-weights are non-negative. 
%\rajmargin{I thought digraph means directed graph.}
%\rajmargin{See, the third word in the paragraph above Definition 1}
\begin{definition}\label{def:strong PF}(\textbf{\textit{strong complex  PF property\cite{varga2012}}}) A matrix \( M \in \mathbb{C}^{n \times n} \) follows the strong complex Perron-Frobenius (PF) property if its
%it possesses a 
dominant eigenvalue \( \lambda_1 \) is positive, simple, and satisfies \( \lambda_1 > |\lambda_j| \) for all \( j \in \{ 2, 3, \dots, n \} \). Additionally, the corresponding right eigenvector satisfies \( \Re(\mathbf{x}) > \mathbf{0} \).
\end{definition}
\begin{definition}\label{def:real EEP}(\textbf{\textit{real eventually exponentially positive matrix\cite{saxena2024realeventualexponentialpositivity}}}) A matrix $M \in \mathbb{C}^{n\times n}$ is real eventually exponentially positive (rEEP) if there exists $t_{0}$ such that $\Re(e^{Mt})>0$ for all time $t\geq t_{0}$.
 \end{definition}
\begin{definition}\label{def:pseudoinverse}
%\rajmargin{have you defined self-adjoint?}
%A matrix $M^{\dagger} \in \mathbb{C}^{n\times m}$ is said to be 
The pseudoinverse (in the sense of Moore-Penrose) of $M \in \mathbb{C}^{n\times m}$ is denoted by $M^{\dagger}$ and it satisfies \cite{2012moore}: 
\begin{enumerate}[label=(\roman*)]
     \item  $M M^{\dagger}M=M$,\label{eq1:Linv}
     \item $M^{\dagger}M M^{\dagger}=M^{\dagger}$, \label{eq2:Linv}
     \item  $M M^{\dagger} \in\mathbb{C}^{m \times m}$ and $M^{\dagger}M \in \mathbb{C}^{n \times n}$ are self-adjoint. \label{eq3:Linv} 
 \end{enumerate}\end{definition}

We cull some important results (without proof) from our prior work \cite{saxena2024realeventualexponentialpositivity}. First, consider the set of matrices
\begin{align}\label{eq:Q}
    \mathcal{O}:=\{O\in\mathbb{C}^{n\times n};\Re(\mathbf{x})\geq\left | \Im(\mathbf{x}) \right |  , \Re(\mathbf{z})\geq\left | \Im(\mathbf{z}) \right |  \}  
\end{align}
and every $O \in \mathcal{O}$ follows strong complex PF Property.
\begin{theorem}\label{stronglemma}
Consider $M$ and $M^{H} \in \mathbb{C}^{n \times n}$ such that $\Re(\mathbf{x})\geq\left | \Im(\mathbf{x}) \right |, \Re(\mathbf{z})\geq\left | \Im(\mathbf{z}) \right | $ holds where $\mathbf{x}$ and $\mathbf{z}$ are the dominant right and left eigenvectors of $M$, respectively. Then, the following statements are equivalent for some scalar $d\geq 0$: 
%\rajmargin{$\mathcal{Q}$ appears weird. Could you try something?}
\begin{enumerate}[label=(\roman*)]
\item $ (M+dI) \text{ and } (M+dI)^{H} \in \mathcal{O}$,
\item $(M+dI)$ is real eventually positive for all $k\geq k_{0}$,
\item $M$ is rEEP.
\end{enumerate}  
where $\mathcal{O}$ is as defined in Eqn.\eqref{eq:Q}.
\end{theorem}
\begin{lemma}
\label{Thm:L undirected}
Consider an undirected unsigned graph such that $L \in \mathbb{C}^{n\times n}$. Then, the following statements are equivalent:
\begin{enumerate}[label=(\roman*)]
\item $-L$ is rEEP,
\item $-L$ has `0' as a simple eigenvalue.
\end{enumerate}
\end{lemma}
%\rajmargin{First sentences of Lemma 1 and 2 do not make any sense to me!}
\begin{lemma}\label{Thm:L digraph}
Consider an unsigned digraph such that $L\in\mathbb{C}^{n \times n}$ corresponds to a weight-balanced and strongly connected graph. Then, the following statements are equivalent:
\begin{enumerate}[label=(\roman*)]
    \item $-L$ is rEEP,
    \item $-L$ has `0' as a simple eigenvalue.
\end{enumerate}
\end{lemma} 
\section{Psuedoinverse Laplacian flows}\label{sec: pseudo-motivation}
Consider the following continuous-time consensus dynamics for the $i$-th agent, for all $i=\{1,2,\ldots,N\}$: 
\begin{align}\label{eq: pseudo-motivation1}
\dot{x}_i(t)=-\frac{x_i(t)}{N_i}+\dot{x}_j(t), 
\end{align}
where $N_i\leq N$ is the number of neighbors of the $i$-th agent. 
%As a physical motivation, we could think of $x_i(t)$ as a position of the $i$-th vehicle. The dynamics in Eqn.\eqref{eq: pseudo-motivation1} then describe the relative velocity between vehicles $i$ and $j$. 
Aggregating the dynamics in Eqn.\eqref{eq: pseudo-motivation1} across the index $j$ and rearrangement leads to 
%\begin{align}
%N_i\dot{x}_i(t)=-x_i(t)+\sum_{j\sim i}\dot{x}_j(t),  
%\end{align}
%which upon 
\begin{align}
\sum_{j\sim i}(\dot{x}_i(t)-\dot{x}_i(t))=-x_i(t). 
\end{align}
Let $x(t)=[x_1(t),\ldots,x_N(t)]^\top$.
Then the consensus dynamics for all agents are governed by $L\dot{x}(t) = -x(t)$, which, through the Moore-Penrose pseudoinverse formalism introduced earlier, results in the pseudoinverse Laplacian flow system:
\begin{align}\label{eq: pseudo-motivation2}
\dot{x}(t)=-L^\dagger x(t). 
\end{align}
We discuss how pseudoinverse flow systems arise naturally in electrical networks. Consider the impedance network in Fig.\ref{fig:z n/w} with branch (or series) resistances and capacitance and shunt inductance ($\mathcal{L}_\text{ind}$). Employing \emph{Kirchhoff's circuit laws}, we have the balance equation as: \cite{bullo2018lectures}
\begin{align}\label{eq:ohms}
    I(t)=Lv(t),
\end{align}
where $I(t)$ and $v(t)$ are the instantaneous nodal currents and voltages vectors, and $L$ is the complex-valued Laplacian arising from the impedance elements. On the other hand, Eqn.\ref{eq:vol.L}
\begin{align}\label{eq:vol.L}
    v(t)=-\mathcal{L}_\text{ind} \frac{dI}{dt}
\end{align}
establishes the relationship between nodal voltages and voltage drops across the shunt inductances (see Fig.\ref{fig:z n/w}). Combining Eqns.\eqref{eq:ohms}-\eqref{eq:vol.L}, we get the pseudoinverse flow system as 
\begin{align}\label{eq:L pseudoinverse flow}
    \dot{I}(t)=\frac{1}{\mathcal{L}_{ind}}(-L^{\dagger})I(t).
\end{align}
%Ex.\ref{ex:RLexample} based on this concept is provided in the next section; see Fig.\ref{fig:z n/w}. 
 Laplacian pseudoinverses have notable applications in social and power networks \cite{cetinay2016topological,li2017consensus,fontan2023pseudoinverses,van2017pseudoinverse}. %Laplacian pseudoinverses relate phase angles of voltages to injected power within a network in load flow problems in \cite{cetinay2016topological}. Given the better accuracy of AC load flow compared to DC load flow, the results of this study may prove valuable.
With their potential in different areas, they present an intriguing avenue for research. Hence, we discuss some results on pseudoinverses in the following section. %\rajmargin{Rewrite! You have already motivated pseudoinverses in introduction; so just summarize in one sentence.--ok}

We commence with some general results that hold regardless of whether the graph is signed or unsigned. The only condition here is that the Laplacian matrix should have zero row and column sums along with a simple `$0$' eigenvalue. %\rajmargin{Please change the word "invariant" in the Lemma--ok done}
\begin{lemma}\label{lem:invariant null space}
    The set of eigenvectors corresponding to the zero eigenvalue of the Laplacian matrix remains same as of the Laplacian pseudoinverse matrix.
\end{lemma}
\begin{proof}
    For undirected graphs and weight-balanced digraphs, we have $L \mathbb{1}_{n}=\mathbb{0}_{n} \text{ and } \mathbb{1}^T_{n}L=\mathbb{0}^T_{n}.$ Further, from Def.\ref{def:pseudoinverse},
    \begin{equation}\label{eq:selfadjoint1}
    (L^{\dagger}L)^H=L^{\dagger}L,  
    \end{equation}
    \begin{equation}\label{eq:selfadjoint2}
    (L L^{\dagger})^H=L L^{\dagger},  
    \end{equation}
 Using Eqn.\eqref{eq:selfadjoint1},
$ \mathbb{1}^T_{n}L^{\dagger}= \mathbb{1}^T_{n}L^{\dagger}L L^{\dagger}= \mathbb{1}^T_{n}(L^{\dagger}L)^H L^{\dagger},$ %$=\mathbb{1}^T_{n}L^H(L^{\dagger})^H L^{\dagger}$
        \begin{align*}
        \mathbb{0}_n^T (L^{\dagger})^H L^{\dagger}=\mathbb{0}^T_n,
    \end{align*}
    Similarly using Eqn.\eqref{eq:selfadjoint2}, $L^{\dagger}\mathbb{1}_{n}= \mathbb{0}_n.$
    \end{proof}
Thus, if the eigenvectors of $L$ corresponding to the zero eigenvalue lie in $\operatorname{span}\left \{  \mathbb{1}_{n}\right \}$, then the same holds true for the eigenvectors of $L^\dagger$. Next, we derive the following result to state the relation between the Laplacian and pseudoinverse Laplacian flows in reaching consensus. 
%Agreement or consensus protocols usually find applications in multi-agent systems, wireless communication graphs and synchronization of linear systems\cite{dorfler2012synchronization,olfati2007consensus}. 
Further, we show in Ex.\ref{ex:RLexample}, that the current flowing in an impedance circuit stabilizes to a steady-state value following a transient phase as an application of Laplacian pseudoinverse flows.%Instead of the Laplacian matrix, we use the pseudoinverse Laplacian matrix to study the stability of system. 
\begin{theorem}\label{Thm:consensus}
In undirected graphs and weight-balanced digraphs, the Laplacian flows achieve consensus if and only if the pseudoinverse Laplacian flows achieve consensus. %\rajmargin{We have decided to use the statement We assume that -L is rEEP in Thm 2.-- If we assume that -L is rEEP in Thm 2., then iff can't be written anymore.}
 \end{theorem}
\begin{proof}
The Laplacian flows are given as $\dot{x}=-Lx$ for all $L \in  \mathbb{C}^{n \times n}$\cite{saxena2024realeventualexponentialpositivity}. Thus, the solution for Laplacian flows is given by $x(t)=e^{-Lt}x_{0}$, where $x_{0}$ is the initial condition.
If the Laplacian flow achieves consensus, then at $t \to \infty$, the following holds from \cite{saxena2024realeventualexponentialpositivity},
\begin{equation} \label{eq:sol.x}
x(t)= \mathbb{1}_{n} z^{H}x_{0}, 
\end{equation}
where $z^{H}x_{0}$ is constant, and hence $x(t)$ is constant. Now, let us consider the pseudoinverse Laplacian flows in 
Eqn.\eqref{eq: pseudo-motivation2}: $\dot{x}=-L^\dagger x$.
From Corollary \ref{Thm:undirected n/w} and Eqn.\eqref{eq:sol.x}, we conclude that as $t \to \infty$, the state $ x(t)=e^{-L^{\dagger}t}x_{0}$ achieves consensus. 
The reverse direction can also be proved on the same lines.\end{proof} 
Thm.\ref{Thm:consensus} shows that, in a certain class of networks, it is enough to check for consensus in either one of the flows: the Laplacian flow or its pseudoinverse; the other one follows from it. But, does the Laplacian (or the pseudoinverse Laplacian) flow achieve consensus for every network in this class? The answer to this question lies in the rEEP property as discussed next.
%Next, we provide a detailed analysis of these graphs, exploring the rEEP property. 
\section{Pseudoinverses of real eventually exponentially positive Laplacians}\label{sec:PSEUDOINVERSE}
%We use this result to claim the property of rEEP in the Laplacian pseudoinverses as discussed next. 
\subsection{Undirected networks}
In social networks, signed networks are more realistic than unsigned networks. 
We consider a set of signed graphs where the real part is non-negative and the imaginary part can be negative or non-negative. In this case, we evaluate the Laplacian matrix in an analogous way to unsigned networks (given in Sec.\ref{sec:preliminaries}). The motivation for leveraging this set of signed graphs stems from power system networks where the Laplacian matrix corresponds to the network’s admittance matrix. The corresponding signed Laplacian retains the fundamental properties of the former. The eigenvalues of both Laplacian and signed Laplacian matrices lie in the closed RHP. This can be proved using the $Ger\check{s}gorin ~Disk~ Theorem$\cite{bullo2018lectures}. Additionally, the signed Laplacian exhibits zero row and column sums as it is a complex symmetric matrix. 

We explore the rEEP property of Laplacian pseudoinverse matrices for the above-mentioned set of signed graphs. 
\begin{theorem}\label{Thm:signed}
Consider an undirected graph $\mathcal{G}(A)$ which is signed and connected such that $L \text{ and }L^{\dagger} \in \mathbb{C}^{n \times n}$ is complex symmetric, then the following statements are equivalent:
\begin{enumerate}[label=(\roman*)]
\item $L$ is $psd$ of corank 1,
\item -$L$ is rEEP,
\item $L^{\dagger}$ is $psd$ of corank 1,
\item -$L^{\dagger}$ is rEEP. 
\end{enumerate} \end{theorem}
\begin{proof}
$(i) \implies (ii)$: We have, $L$ is $psd$ of corank $1$ which implies $\lambda_1=0 < \lambda_2 \leq \lambda_3....\leq \lambda_n$ are the eigenvalues of $L$. We choose a value of $d \in \mathbb{R}_{>0}$ such that $\rho(B)=d$ and the corresponding right eigenvector lies in the $\operatorname{span}\left \{  \mathbb{1}_{n}\right \}$. Then, we consider $B=dI-L$ which implies $B\mathbb{1}_{n}=d\mathbb{1}_{n}-L\mathbb{1}_{n}$. Therefore, $B$ has a dominant eigenvalue $d$ with $x \in \operatorname{span}\left \{  \mathbb{1}_{n}\right \}$. Hence, $B$ follows strong complex PF property and is real eventually positive (rEP). Using Thm.\ref{stronglemma}, $\Re(B^k)>0$ implies that $-L$ is rEEP. 
    
$(ii) \implies (i)$: Conversely, $-L$ is rEEP. From the previous part, we consider $B=dI-L$. Further, using Thm.\ref{stronglemma}, we know that $B$ is rEP, which is equivalent to $B$ having a simple dominant eigenvalue, $\rho(B)=d \text{ and } x \in \operatorname{span}\left \{  \mathbb{1}_{n}\right \}$. Reversing the construction of $L$, it must have a simple `0' eigenvalue and the remaining eigenvalues lie in the open RHP. Hence, $L$ is $psd$ of corank $1$. 

$(ii) \implies (iii)$: We have already shown that if $-L$ is rEEP, it has a simple `0' eigenvalue equivalently. Then, we have the spectrum of $L^{\dagger}$ given in Eqn.\eqref{eq:spec(Linv)},
\begin{equation}\label{eq:spec(Linv)}
spec(-L^{\dagger})=\left \{ 0, -1/\lambda _{2},-1/\lambda _{3},...., -1/\lambda _{n}\right \},
\end{equation}
where $\lambda _{1},\lambda _{2},...,\lambda _{n}$ are the eigenvalues of $L$ \cite{fontan2023pseudoinverses}. Using Lemma~\ref{lem:invariant null space}, $x,y^H \in \operatorname{span}\left \{  \mathbb{1}_{n}\right \}$. By the inherent properties of $L$, we know that all the eigenvalues lie in the closed RHP. Let $\lambda_j$ be a non-zero eigenvalue which belongs to the $\operatorname{spec}(L)$. Then, all the non-zero eigenvalues of $L^{\dagger}$ have positive real parts.
Therefore, $L^\dagger$ has all the eigenvalues in the closed RHP. Hence, $L^{\dagger}$ is $psd$ of corank 1.
     
$(iii) \implies (iv)$: We have, $L^{\dagger}$ is $psd$ of corank 1. Using the Jordan canonical form of $L^{\dagger}$, 
\begin{align}\label{eq:Linv jordan}
-L^{\dagger}=\begin{bmatrix}
 x^{(1)}|& 
X_{n,n-1}\end{bmatrix} \begin{bmatrix} -\lambda_{1} & 0 \\ 0 & -D^{-1}_{n-1,n-1} \end{bmatrix} \begin{bmatrix}
y^{(1)^{H}}\\Y_{n-1,n} \end{bmatrix},
\end{align}
Substituting $\lambda_{1}=0$, we get
\begin{align*}
-L^{\dagger}=\begin{bmatrix}
 x^{(1)}|& 
X_{n,n-1}\end{bmatrix} \begin{bmatrix} 0 & 0 \\ 0 & -D^{-1} \end{bmatrix} \begin{bmatrix}
y^{(1)^{H}}\\Y_{n-1,n} \end{bmatrix},
\end{align*}
where $D$ is the diagonal matrix containing eigenvalues of $L$ except `$0$' eigenvalue. Then, taking exponential of the negated Laplacian matrix, at a larger time instant, we have the final expression for matrix exponential as: $ \underset{t \rightarrow \infty}{lim} ~ e^{-L^{\dagger}t}= \alpha x^{(1)}y^{(1)^{H}}$. Since $x^{(1)} \text{ and } y^{(1)^{{H}}} \in \operatorname{span}\left \{  \mathbb{1}_{n}\right \}$, the product of the eigenvectors is always real and positive. Hence, $-L^{\dagger}$ is rEEP.  

 $(iv) \implies (iii)$: The Laplacian pseudoinverse matrix satisfies the two inherent properties of $L$. The first property is given in Lemma~\ref{lem:invariant null space}, and the second one states that all eigenvalues, excluding the zero eigenvalue, lie in open RHP (shown in $(ii) \implies (iii)$).
 Hence, it can be proven on the same lines as proved in $(ii) \implies (i)$ direction. 
\end{proof} 
\begin{corollary}\label{Thm:undirected n/w}
Consider an undirected graph which is unsigned and connected such that $L \text{ and }L^{\dagger} \in \mathbb{C}^{n \times n}$. Then, the following statements are equivalent:
\begin{enumerate}[label=(\roman*)]
\item $-L$ is rEEP,
\item $-L^{\dagger}$ is $psd$ of corank 1,
\item $-L^{\dagger}$ is rEEP.
\end{enumerate}
\end{corollary}
\begin{proof}
Using Lemma~\ref{Thm:L undirected}, $-L$ being rEEP is equivalent to $L$ having a simple `0' eigenvalue. Since the spectrum of Laplacian pseudoinverse matrix is given in Eqn.\eqref{eq:spec(Linv)} and Thm.\ref{Thm:signed} already proves that the property of rEEP is conserved under the operation of pseudoinverse. Thus, the inherent properties of $L \text{ and }L^{\dagger}$ remain same as in Thm.\ref{Thm:signed}. Hence, this can be proved on the same lines as in Thm.\ref{Thm:signed}.
\end{proof}
%We now explore unsigned weight-balanced digraphs in the next subsection.
\subsection{Unsigned directed networks}
We assume the digraph $\mathcal{G}(A)$ to be strongly connected and weight-balanced. The rEEP property, which enables consensus, is a consequence of the network's connectivity. The assumption of weight-balance aids in proving the rEEP property, as the left eigenvector corresponding to the `$0$' eigenvalue in a non weight-balanced digraph is indeterminate in the case of a complex-valued digraph. 
\begin{theorem}\label{Thm3:unsigned digraph}
 Consider an unsigned digraph $\mathcal{G}(A)$ such that $L \text{ and }L^{\dagger} \in \mathbb{C}^{n \times n}$ corresponds to a  strongly connected and weight-balanced digraph. Then, the following statements are equivalent:
\begin{enumerate}[label=(\roman*)]
\item -$L$ is rEEP.
\item -$L^{\dagger}$ has a simple `0' eigenvalue.
\item -$L^{\dagger}$ is rEEP.
\item $(L^{\dagger})_{s}$ is $psd$ of corank 1.
\end{enumerate} \end{theorem}
\begin{proof} 
$(i)\implies(ii)$: From Lemma~\ref{Thm:L digraph}, if $-L$ is rEEP, then it equivalently has a simple `$0$' eigenvalue. Hence, we have the spectrum of $L^\dagger$ given in Eqn.\eqref{eq:spec(Linv)}. Consequently, $L^{\dagger}$ has a simple `0' eigenvalue. 

$(ii)\implies(i)$: We have, $L^{\dagger}$ with a simple `0' eigenvalue. Following Eqn.\eqref{eq:spec(Linv)}, $L$ has a simple `0' eigenvalue. Applying Lemma~\ref{Thm:L digraph}, $-L$ is rEEP.

$(ii)\implies(iii)$: This can be proved using Jordan canonical form. Substituting $\lambda_{1}=0$ in Eqn.\eqref{eq:Linv jordan}, we obtain $\underset{t \rightarrow \infty}{lim} ~ e^{-L^{\dagger}t}= \alpha x^{(1)}y^{(1)^{H}}$. The eigenvectors corresponding to zero eigenvalue are same (using Lemma~\ref{lem:invariant null space}). Hence, $-L^{\dagger}$ is rEEP. The reverse direction can also be proved as in Thm.\ref{Thm:signed}.

\looseness=-1
$(ii)\implies(iv)$: The symmetric part of $L^\dagger$, $(L^{\dagger})_{s}$ follows ${\Re}\left(\lambda_i(L^{\dagger})\right) {=} \lambda_i\left((L^{\dagger})_{s}\right) \text{ which implies }  L^{\dagger}_s$ is $psd$ of corank 1. 

$(iv)\implies(ii)$: Since $L^{\dagger}_{s}$ is $psd$ of corank $1$ and it follows ${\Re}\left(\lambda_i(L^{\dagger})\right) {=} \lambda_j\left((L^{\dagger})_{s}\right)$, it implies that all the eigenvalues of $L^{\dagger}$ lie in open RHP. Hence, $-L^{\dagger}$ has a simple `$0$' eigenvalue.
\end{proof} Moreover, the assumption of weight-balance guarantees average consensus, where the final convergence value is achieved exactly as an average of all initial opinions. %The final value of convergence is real-valued asymptotically.  
Additionally, the property of rEEP implies that all agents in a network have equal importance in the network. 
Next, we provide a consequence of rEEP property. 
\begin{corollary}\label{rem:consensus}
Through the results presented in this section, it follows that if $-L$ is rEEP, then the Laplacian and the pseudoinverse Laplacian flows achieve consensus for every undirected graphs and weight-balanced unsigned digraphs.
\end{corollary}
It follows from Corollary \ref{rem:consensus} that in general, rEEP is only a sufficient condition for consensus. A network of agents may converge to consensus even when $-L$ is not rEEP (an example is shown in Fig.\ref{fig:counter}). These cases may be explored in the future.
 \begin{figure}[H]
    \begin{subfigure}{0.2\textwidth}
        \centering
        \includegraphics[scale=0.17]{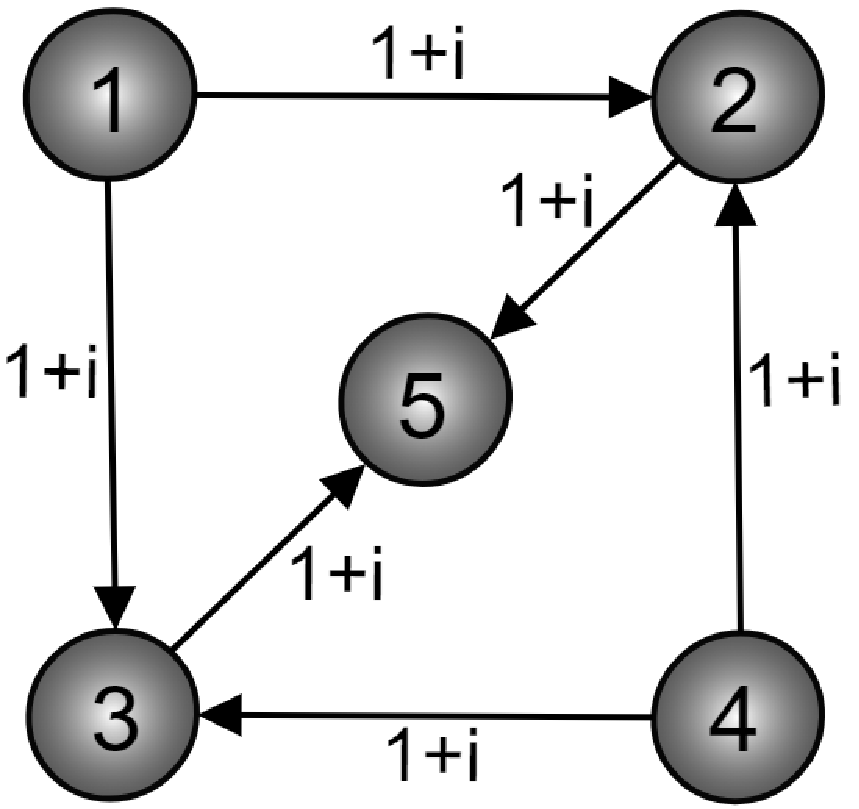}
        \caption{\small A weakly connected digraph}
        \label{fig:weaklydigraph}
        \end{subfigure}
   % \hspace{0.27cm} % Adjust the space between the two figures
    \begin{subfigure}{0.2\textwidth}
        \centering
        \includegraphics[scale=0.556]{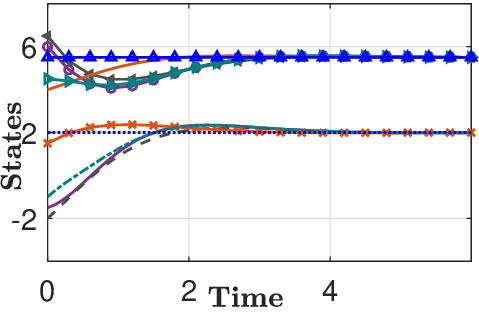}
        \caption{Laplacian flow}
        \label{fig:Lflowcounter}
          \end{subfigure}     
    \caption{ \small For the given network, $-L$ is not rEEP, yet consensus is achieved. The trajectories above and below in Fig.\ref{fig:Lflowcounter} denote the real and imaginary parts of the states, respectively.} %The exact values of $e^{-Lt}$ are not provided here in the interest of space.}
    \label{fig:counter}
\end{figure}
Now, we present the following numerical examples to validate the aforementioned results. Note that, for brevity, all matrices in the examples have been rounded to two decimal places.
\begin{example}\label{ex:unsigned}\emph{(Synthetic network)}
Here, the given unsigned digraph is strongly connected and weight-balanced.
  \begin{figure}[H]
        \centering
        \includegraphics[scale=0.165]{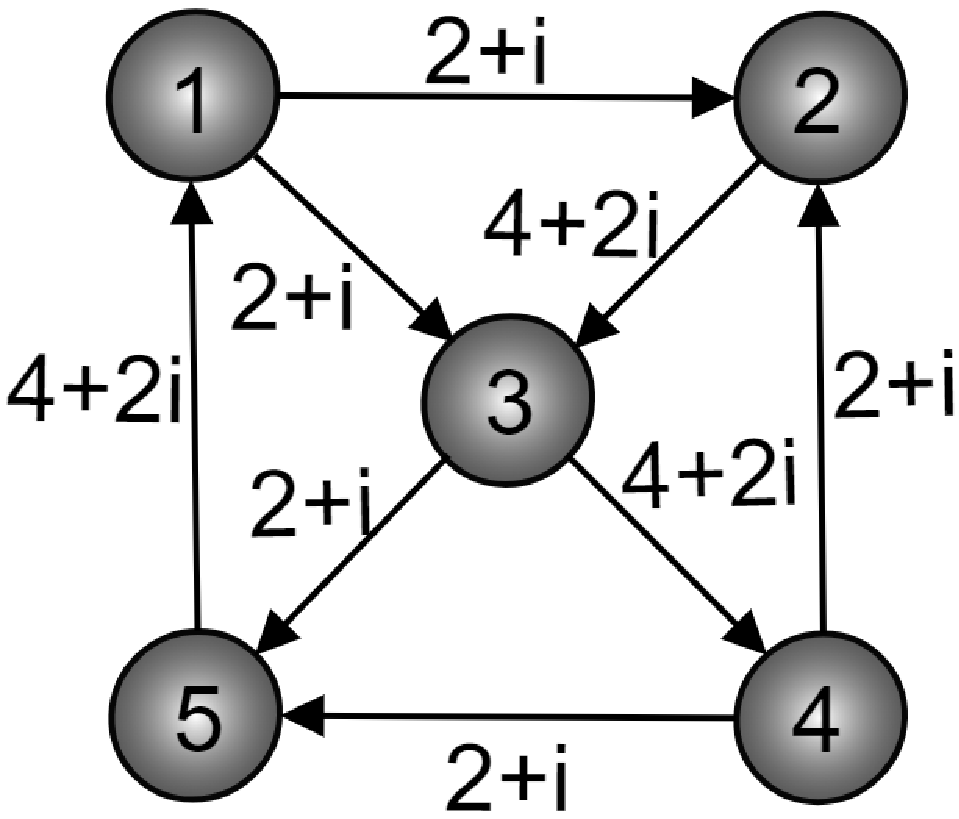}
        \caption{Strongly connected weight-balanced digraph}
        \label{fig:wt.balanced digraph}
\end{figure}%\vspace{-10mm}
The matrix exponential of the negated Laplacian matrix is
\begin{equation*}
e^{-L_{1}t} = \renewcommand{\arraystretch}{1} % Adjust vertical spacing
\setlength{\arraycolsep}{2pt}    % Adjust column spacing
\scriptsize
\begin{bmatrix}
 0.2 + 0.01i & 0.2 & 0.19 &0.21& 0.2-0.01i \\ 
0.22 & 0.2 & 0.2+0.01i & 0.2-0.01i & 0.19-0.01i\\ 
0.2-0.01i &0.2 &  0.21 &0.19& 0.2+0.01i \\
0.18  & 0.2& 0.2 - 0.01i & 0.2+0.01i & 0.22 \\
 0.2 &  0.2 &  0.2 &  0.2&  0.2 
\end{bmatrix}
\normalsize
\end{equation*}
%\looseness=-1
 Here, $\operatorname{spec}(L_{1})=\left \{ 0,3.67+5.14i, 5+5i,6.32-0.14i,7+i \right \}$. $\text{For } d \geq 6$, we note that $B=dI-L$ is rEEP and $B \in \mathcal{O}$. So, $-L_{1}$ has a simple `0' eigenvalue. Further, $-L_1$ is equivalently rEEP.
The eigenvalues of the Laplacian pseudoinverse lie in the closed RHP as
$\operatorname{spec}(L_{1}^{\dagger}) = \left\{
0, 0.091 - 0.128i, 0.1 - 0.1i, 0.16 + 0.003i, 0.14 - 0.02i \right\}$. Then, the matrix exponential for the Laplacian pseudoinverse at $t=13$ is given as: 
\begin{equation*}
\renewcommand{\arraystretch}{1} % Adjust vertical spacing
\setlength{\arraycolsep}{2pt}    % Adjust column spacing
\scriptsize
\begin{bmatrix}
0.23+0.11i  & 0.13-0.05i & 0.13-0.07i & 0.25+0.01i & 0.26-0.01i\\ 
0.24 + 0.04i & 0.25 + 0.1i & 0.01-0.08i & 0.26-0.01i & 0.25-0.05i\\ 
0.26 - 0.01i &  0.24 - 0.01i &  0.36+ 0.16i & 0.13- 0.07i  & 0.01- 0.08i \\
0.02- 0.13i &  0.24- 0.01i &  0.26- 0.01i & 0.23 + 0.11i  & 0.24 + 0.04i \\
0.24 - 0.01i  & 0.13- 0.05i &   0.24- 0.01i  & 0.13- 0.05i & 0.25+0.1i 
\end{bmatrix}
\normalsize
\end{equation*}
\looseness=-1 Thus, we observe that both the negated Laplacian and Laplacian pseudoinverse matrices are rEEP (as proved in Thm.\ref{Thm3:unsigned digraph}). 
\end{example}
\begin{example}\label{ex:signed}\emph{(Real-world power system network)}
We consider signed Laplacian matrix corresponding to a case 10ba network imported from MATPOWER software. 
\begin{figure}[H]
    \centering
    \begin{subfigure}{0.2\textwidth}
    \centering
        \includegraphics[scale=0.61]{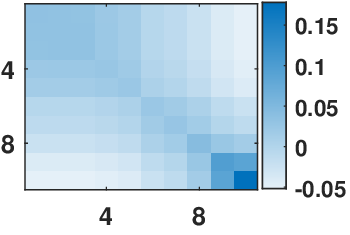}
        \caption{$\Re(L^{\dagger})$}
        \label{fig:Linv_realcase10}
        \end{subfigure}
    \hspace{0.34cm} % Adjust the space between the two figures
    \begin{subfigure}{0.2\textwidth}
    \centering
\includegraphics[scale=0.61]{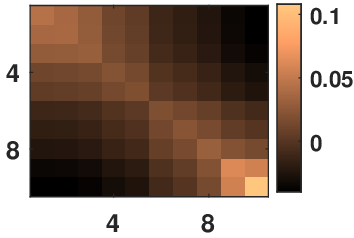}
        \caption{$\Im(L^{\dagger})$}
        \label{fig:Linv_imagcase10}
          \end{subfigure}     
    \caption{\small Heatmaps showing the matrix structure of $L^\dagger$ for real and imaginary parts respectively in an IEEE 10 bus network (the number here denotes the total nodes)}
    \label{figs:case10 realEEP}
\end{figure}
\begin{figure}[H]
    \centering
    \begin{subfigure}{0.2\textwidth}
    \centering
        \includegraphics[scale=0.825]{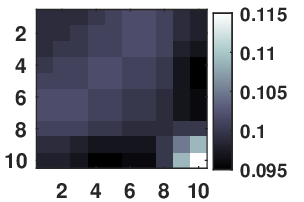}
        \raggedright
        \caption{$\Re(e^{-Lt})$}
        \label{fig:L_realeep case10}
        \end{subfigure}
    \hspace{0.56cm} % Adjust the space between the two figures
    \begin{subfigure}{0.2\textwidth}
    \centering
\includegraphics[scale=0.8]{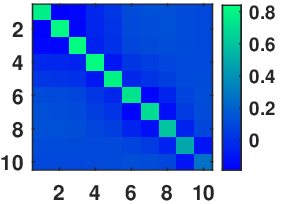}
\raggedright
        \caption{$\Re(e^{-L^{\dagger}t}$)}
        \label{fig:Linv_realeep case10}
          \end{subfigure}     
    \caption{\small The matrix structure of $\Re(e^{-Lt})$ and $\Re(e^{-L^{\dagger}t}$) at $t=1$ and $t=15$ respectively. Values on the colorbar represent the entire range that encompasses all the values of a matrix, and all of these values are positive implying rEEP property.}
    \label{figs:case10 realEEP}
\end{figure}
The matrix structure of $L$ corresponding to the 10-bus network closely resembles to that of the $L^\dagger$. The exact heat maps have been omitted for brevity.

Since the spectrum of the Laplacian matrix corresponding to the 10-bus network has a simple zero eigenvalue, the negated Laplacian matrix as well as negated Laplacian pseudoinverse is rEEP.
Fig.\ref{figs:case10 realEEP} illustrates the approximated values of the matrix exponential of the negated Laplacian matrices, as well as their Laplacian pseudoinverses.
Thus, it is evident from Fig.\ref{figs:case10 realEEP} that the rEEP property is preserved under the operation of pseudoinverse for Laplacian matrices, as proved in Thm.\ref{Thm:signed}. 
\end{example}
\begin{example}\label{ex:RLexample}\emph{(Impedance network)}
%\looseness=-1
Here, we consider an impedance network with shunt inductors connecting each node to ground in Fig.\ref{fig:z n/w}. 
The following network is connected, and the spectrum of corresponding $L$ contains a simple `0' eigenvalue. This network corresponds to a signed Laplacian matrix. 
\tikzset{every picture/.style={line width=0.56pt}} %set default line width to 0.75pt        
\begin{tikzpicture}[x=0.5pt,y=0.5pt,yscale=-1,xscale=1]\label{fig:z n/w}
%uncomment if require: \path (0,293); %set diagram left start at 0, and has height of 293
%Shape: Inductor (Air Core) [id:dp7736825365435145] 
\draw  [line width=1.5]  (215,41) -- (229.4,41) .. controls (229.61,34.96) and (231.61,29.79) .. (234.43,27.99) .. controls (237.26,26.18) and (240.35,28.1) .. (242.2,32.83) .. controls (243.63,36.51) and (244.21,41.28) .. (243.8,45.9) .. controls (243.8,47.71) and (243.08,49.17) .. (242.2,49.17) .. controls (241.32,49.17) and (240.6,47.71) .. (240.6,45.9) .. controls (240.19,41.28) and (240.77,36.51) .. (242.2,32.83) .. controls (243.86,28.9) and (246.18,26.67) .. (248.6,26.67) .. controls (251.02,26.67) and (253.34,28.9) .. (255,32.83) .. controls (256.43,36.51) and (257.01,41.28) .. (256.6,45.9) .. controls (256.6,47.71) and (255.88,49.17) .. (255,49.17) .. controls (254.12,49.17) and (253.4,47.71) .. (253.4,45.9) .. controls (252.99,41.28) and (253.57,36.51) .. (255,32.83) .. controls (256.66,28.9) and (258.98,26.67) .. (261.4,26.67) .. controls (263.82,26.67) and (266.14,28.9) .. (267.8,32.83) .. controls (269.23,36.51) and (269.81,41.28) .. (269.4,45.9) .. controls (269.4,47.71) and (268.68,49.17) .. (267.8,49.17) .. controls (266.92,49.17) and (266.2,47.71) .. (266.2,45.9) .. controls (265.79,41.28) and (266.37,36.51) .. (267.8,32.83) .. controls (269.65,28.1) and (272.74,26.18) .. (275.57,27.99) .. controls (278.39,29.79) and (280.39,34.96) .. (280.6,41) -- (295,41) ;
%Shape: Resistor [id:dp8932617178755298] 
\draw  [color={rgb, 255:red, 25; green, 23; blue, 25 }  ,draw opacity=1 ][line width=1.5]  (135,41) -- (149.4,41) -- (152.6,31.79) -- (159,50.21) -- (165.4,31.79) -- (171.8,50.21) -- (178.2,31.79) -- (184.6,50.21) -- (191,31.79) -- (197.4,50.21) -- (200.6,41) -- (215,41) ;
%Straight Lines [id:da3981351040670884] 
\draw [line width=1.5]    (295,41) -- (358,41) ;
%Straight Lines [id:da9516053637082629] 
\draw [line width=1.5]    (106.09,41.4) -- (135,41) ;
%Shape: Capacitor [id:dp5121653660463716] 
\draw  [line width=1.5]  (106.91,121.4) -- (107.52,157.39) (122.9,165.13) -- (92.41,165.65) (122.77,157.13) -- (92.27,157.65) (107.65,165.39) -- (108.26,201.38) ;
%Straight Lines [id:da7874077817488518] 
\draw [line width=1.5]    (108.26,201.38) -- (441,200.97) ;
%Straight Lines [id:da10815050516353542] 
\draw    (438,41) -- (478.09,41.26) ;
%Straight Lines [id:da31771864654395876] 
\draw    (465.09,121.26) -- (491.09,121.26) ;
%Straight Lines [id:da2820610633924552] 
\draw    (471.09,129.26) -- (486.09,129.26) ;
%Straight Lines [id:da9497662212018936] 
\draw [line width=1.5]    (106.09,41.4) -- (57.2,42.28) ;
%Straight Lines [id:da631094298990581] 
\draw    (43.99,122.24) -- (69.99,122.31) ;
%Straight Lines [id:da15888672549522487] 
\draw    (49.97,130.26) -- (64.97,130.3) ;
%Straight Lines [id:da2841408564745773] 
\draw    (284.09,280.54) -- (310.09,280.54) ;
%Straight Lines [id:da3689060528511938] 
\draw    (290.09,288.54) -- (305.09,288.54) ;
%Straight Lines [id:da7427234910864196] 
\draw [line width=0.75]    (321.09,121.26) -- (347.09,121.26) ;
%Straight Lines [id:da3835892629666997] 
\draw [line width=0.75]    (327.09,129.26) -- (342.09,129.26) ;
%Shape: Resistor [id:dp6410403107594178] 
\draw  [line width=1.5]  (106.91,41.4) -- (106.91,55.8) -- (115.89,59) -- (97.93,65.4) -- (115.89,71.8) -- (97.93,78.2) -- (115.89,84.6) -- (97.93,91) -- (115.89,97.4) -- (97.93,103.8) -- (106.91,107) -- (106.91,121.4) ;
%Shape: Resistor [id:dp24021281653800242] 
\draw  [line width=1.5]  (295.82,41) -- (295.82,55.4) -- (305.26,58.6) -- (286.39,65) -- (305.26,71.4) -- (286.39,77.8) -- (305.26,84.2) -- (286.39,90.6) -- (305.26,97) -- (286.39,103.4) -- (295.82,106.6) -- (295.82,121) ;
%Shape: Capacitor [id:dp10827462830944756] 
\draw  [line width=1.5]  (295.82,121) -- (296.43,156.99) (311.82,164.73) -- (281.32,165.25) (311.68,156.73) -- (281.19,157.25) (296.57,164.99) -- (297.18,200.98) ;
%Shape: Inductor (Air Core) [id:dp9033970718475794] 
\draw  [line width=0.75]  (58.41,122.31) -- (58.29,107.91) .. controls (52.24,107.75) and (47.06,105.8) .. (45.23,102.98) .. controls (43.4,100.17) and (45.3,97.07) .. (50.01,95.18) .. controls (53.68,93.72) and (58.44,93.1) .. (63.07,93.47) .. controls (64.88,93.45) and (66.35,94.16) .. (66.35,95.04) .. controls (66.36,95.93) and (64.9,96.65) .. (63.1,96.67) .. controls (58.48,97.12) and (53.71,96.58) .. (50.01,95.18) .. controls (46.07,93.55) and (43.82,91.25) .. (43.8,88.83) .. controls (43.78,86.41) and (45.99,84.08) .. (49.9,82.38) .. controls (53.57,80.92) and (58.33,80.3) .. (62.96,80.67) .. controls (64.77,80.66) and (66.24,81.36) .. (66.25,82.24) .. controls (66.25,83.13) and (64.8,83.85) .. (62.99,83.87) .. controls (58.37,84.32) and (53.6,83.78) .. (49.9,82.38) .. controls (45.96,80.75) and (43.71,78.45) .. (43.69,76.03) .. controls (43.67,73.61) and (45.88,71.28) .. (49.79,69.58) .. controls (53.47,68.12) and (58.23,67.5) .. (62.86,67.87) .. controls (64.66,67.86) and (66.13,68.56) .. (66.14,69.44) .. controls (66.15,70.33) and (64.69,71.06) .. (62.88,71.07) .. controls (58.26,71.52) and (53.49,70.98) .. (49.79,69.58) .. controls (45.05,67.77) and (43.11,64.7) .. (44.89,61.86) .. controls (46.67,59.01) and (51.82,56.97) .. (57.86,56.71) -- (57.74,42.31) ;
%Shape: Inductor (Air Core) [id:dp8323072650055503] 
\draw   (478.09,41.26) -- (478.25,55.66) .. controls (484.3,55.79) and (489.49,57.73) .. (491.33,60.54) .. controls (493.17,63.35) and (491.28,66.45) .. (486.58,68.36) .. controls (482.91,69.83) and (478.15,70.47) .. (473.52,70.11) .. controls (471.71,70.13) and (470.24,69.43) .. (470.23,68.55) .. controls (470.22,67.67) and (471.68,66.93) .. (473.48,66.91) .. controls (478.1,66.45) and (482.87,66.97) .. (486.58,68.36) .. controls (490.52,69.98) and (492.78,72.27) .. (492.8,74.69) .. controls (492.83,77.11) and (490.63,79.45) .. (486.73,81.16) .. controls (483.06,82.63) and (478.3,83.27) .. (473.67,82.91) .. controls (471.86,82.93) and (470.39,82.23) .. (470.38,81.35) .. controls (470.37,80.47) and (471.83,79.73) .. (473.63,79.71) .. controls (478.25,79.25) and (483.02,79.77) .. (486.73,81.16) .. controls (490.67,82.78) and (492.93,85.06) .. (492.95,87.49) .. controls (492.98,89.91) and (490.78,92.25) .. (486.88,93.96) .. controls (483.21,95.43) and (478.45,96.07) .. (473.82,95.71) .. controls (472.01,95.73) and (470.54,95.03) .. (470.53,94.15) .. controls (470.52,93.27) and (471.98,92.53) .. (473.78,92.51) .. controls (478.4,92.04) and (483.17,92.57) .. (486.88,93.96) .. controls (491.62,95.76) and (493.58,98.82) .. (491.81,101.67) .. controls (490.03,104.52) and (484.89,106.57) .. (478.85,106.85) -- (479.02,121.25) ;
%Shape: Inductor (Air Core) [id:dp9881629394723574] 
\draw  [line width=0.75]  (332.75,41.27) -- (332.99,55.67) .. controls (339.03,55.77) and (344.23,57.68) .. (346.08,60.48) .. controls (347.94,63.28) and (346.07,66.39) .. (341.37,68.33) .. controls (337.71,69.82) and (332.96,70.48) .. (328.32,70.15) .. controls (326.52,70.18) and (325.04,69.48) .. (325.03,68.6) .. controls (325.01,67.72) and (326.47,66.98) .. (328.27,66.95) .. controls (332.89,66.46) and (337.66,66.96) .. (341.37,68.33) .. controls (345.33,69.92) and (347.59,72.2) .. (347.63,74.62) .. controls (347.67,77.05) and (345.49,79.4) .. (341.59,81.13) .. controls (337.93,82.62) and (333.17,83.28) .. (328.54,82.94) .. controls (326.73,82.97) and (325.26,82.28) .. (325.24,81.4) .. controls (325.23,80.52) and (326.68,79.78) .. (328.49,79.75) .. controls (333.1,79.26) and (337.88,79.76) .. (341.59,81.13) .. controls (345.54,82.72) and (347.81,85) .. (347.85,87.42) .. controls (347.89,89.84) and (345.7,92.2) .. (341.8,93.92) .. controls (338.14,95.41) and (333.39,96.08) .. (328.75,95.74) .. controls (326.95,95.77) and (325.47,95.08) .. (325.46,94.2) .. controls (325.44,93.31) and (326.9,92.57) .. (328.7,92.54) .. controls (333.32,92.05) and (338.09,92.56) .. (341.8,93.92) .. controls (346.56,95.7) and (348.53,98.75) .. (346.77,101.61) .. controls (345.01,104.47) and (339.88,106.55) .. (333.84,106.86) -- (334.09,121.26) ;
%Shape: Resistor [id:dp2506906911565403] 
\draw  [line width=1.5]  (358,41) -- (372.4,41) -- (375.6,32.29) -- (382,49.71) -- (388.4,32.29) -- (394.8,49.71) -- (401.2,32.29) -- (407.6,49.71) -- (414,32.29) -- (420.4,49.71) -- (423.6,41) -- (438,41) ;
%Shape: Resistor [id:dp9369322875374124] 
\draw  [line width=1.5]  (439.35,120.99) -- (439.65,135.38) -- (448.86,138.4) -- (430.71,145.17) -- (449.12,151.19) -- (430.97,157.97) -- (449.39,163.99) -- (431.23,170.76) -- (449.65,176.79) -- (431.5,183.56) -- (440.71,186.57) -- (441,200.97) ;
%Shape: Capacitor [id:dp1254492183891318] 
\draw  [line width=1.5]  (438,41) -- (438.61,76.99) (453.99,84.74) -- (423.5,85.25) (453.86,76.74) -- (423.36,77.25) (438.74,84.99) -- (439.35,120.99) ;
%Shape: Inductor (Air Core) [id:dp11196472040321459] 
\draw   (297.18,200.98) -- (297.35,215.38) .. controls (303.39,215.52) and (308.58,217.46) .. (310.42,220.27) .. controls (312.26,223.07) and (310.37,226.18) .. (305.67,228.09) .. controls (302,229.56) and (297.24,230.2) .. (292.61,229.84) .. controls (290.81,229.86) and (289.33,229.16) .. (289.32,228.28) .. controls (289.31,227.4) and (290.77,226.66) .. (292.57,226.64) .. controls (297.2,226.17) and (301.97,226.7) .. (305.67,228.09) .. controls (309.62,229.7) and (311.87,231.99) .. (311.9,234.41) .. controls (311.92,236.84) and (309.73,239.18) .. (305.82,240.89) .. controls (302.15,242.36) and (297.39,243) .. (292.76,242.64) .. controls (290.96,242.66) and (289.48,241.96) .. (289.47,241.08) .. controls (289.46,240.19) and (290.92,239.46) .. (292.72,239.44) .. controls (297.34,238.97) and (302.12,239.5) .. (305.82,240.89) .. controls (309.76,242.5) and (312.02,244.79) .. (312.05,247.21) .. controls (312.07,249.64) and (309.88,251.98) .. (305.97,253.69) .. controls (302.3,255.16) and (297.54,255.8) .. (292.91,255.44) .. controls (291.11,255.46) and (289.63,254.76) .. (289.62,253.88) .. controls (289.61,252.99) and (291.07,252.26) .. (292.87,252.24) .. controls (297.49,251.77) and (302.27,252.3) .. (305.97,253.69) .. controls (310.72,255.48) and (312.67,258.54) .. (310.9,261.39) .. controls (309.13,264.24) and (303.98,266.3) .. (297.95,266.58) -- (298.11,280.98) ;
%Shape: Circle [id:dp2773648599813803] 
\draw  [fill={rgb, 255:red, 144; green, 19; blue, 254 }  ,fill opacity=1 ] (102.34,41.4) .. controls (102.34,38.87) and (104.39,36.83) .. (106.91,36.83) .. controls (109.43,36.83) and (111.47,38.87) .. (111.47,41.4) .. controls (111.47,43.92) and (109.43,45.96) .. (106.91,45.96) .. controls (104.39,45.96) and (102.34,43.92) .. (102.34,41.4) -- cycle ;
%Shape: Circle [id:dp6787831143327525] 
\draw  [fill={rgb, 255:red, 144; green, 19; blue, 254 }  ,fill opacity=1 ] (290.43,41) .. controls (290.43,38.48) and (292.48,36.43) .. (295,36.43) .. controls (297.52,36.43) and (299.57,38.48) .. (299.57,41) .. controls (299.57,43.52) and (297.52,45.57) .. (295,45.57) .. controls (292.48,45.57) and (290.43,43.52) .. (290.43,41) -- cycle ;
%Shape: Circle [id:dp01969109924682355] 
\draw  [fill={rgb, 255:red, 144; green, 19; blue, 254 }  ,fill opacity=1 ] (292.61,200.98) .. controls (292.61,198.46) and (294.66,196.42) .. (297.18,196.42) .. controls (299.7,196.42) and (301.74,198.46) .. (301.74,200.98) .. controls (301.74,203.51) and (299.7,205.55) .. (297.18,205.55) .. controls (294.66,205.55) and (292.61,203.51) .. (292.61,200.98) -- cycle ;
%Shape: Circle [id:dp5907111807905958] 
\draw  [fill={rgb, 255:red, 144; green, 19; blue, 254 }  ,fill opacity=1 ] (433.43,41) .. controls (433.43,38.48) and (435.48,36.43) .. (438,36.43) .. controls (440.52,36.43) and (442.57,38.48) .. (442.57,41) .. controls (442.57,43.52) and (440.52,45.57) .. (438,45.57) .. controls (435.48,45.57) and (433.43,43.52) .. (433.43,41) -- cycle ;
% Text Node
%\draw (130,71) node [anchor=north west][inner sep=0.75pt]  [font=\large] [align=left] {{$C_{\text{1}}$}};
% Text Node
\draw (121.04,74.8) node [anchor=north west][inner sep=0.75pt]   [align=left] {{$R_{{1}}$}};
 Text Node
%\draw (176,4) node [anchor=north west][inner sep=0.35pt]  [font=\large] [align=left] {\textsubscript{{\fontfamily{ptm}\selectfont 1}}};
% Text Node
\draw (166.4,15.49) node [inner sep=0.55pt]   [align=left] {{$R_2$}};
% Text Node
%\draw (134,148) node [anchor=north west][inner sep=0.75pt]  [font=\large] [align=left] {{\fontfamily{ptm}\selectfont \textsubscript{1}}};
% Text Node
\draw (124.4,152.59) node [anchor=north west][inner sep=0.75pt]   [align=left] {$C_{1}$};
% Text Node
%\draw (271,148) node [anchor=north west][inner sep=0.75pt]  [font=\large] [align=left] {\textsubscript{{\fontfamily{ptm}\selectfont 2}}};
% Text Node
\draw (254.4,152.59) node [anchor=north west][inner sep=0.75pt]   [align=left] {$C_2$};
% Text Node
%\draw (409,70) node [anchor=north west][inner sep=0.85pt]  [font=\large] [align=left] {\textsubscript{{\fontfamily{ptm}\selectfont 3}}};
% Text Node
\draw (399.4,73.59) node [anchor=north west][inner sep=0.75pt]   [align=left] {$C_3$};
% Text Node
%\draw (275,69) node [anchor=north west][inner sep=0.75pt]  [font=\large] [align=left] {\textsubscript{{\fontfamily{ptm}\selectfont 3}}};
% Text Node
\draw (254.4,71.59) node [anchor=north west][inner sep=0.75pt]   [align=left] {$R_{{3}}$};
% Text Node
%\draw (394,4) node [anchor=north west][inner sep=0.75pt]  [font=\large] [align=left] {$C_{\text{1}}$};
% Text Node
\draw (382.4,7.59) node [anchor=north west][inner sep=0.75pt]   [align=left] {$R_4$};
% Text Node
\draw (244.4,15.59) node [inner sep=1.75pt]   [align=left] {$L_1$};
% Text Node
%\draw (251,2) node [anchor=north west][inner sep=0.75pt]  [font=\large] [align=left] {\textsubscript{{\fontfamily{ptm}\selectfont 1}}};
% Text Node
%\draw (414,141) node [anchor=north west][inner sep=0.75pt] [font=\large] [align=left] {$R_5$};
% Text Node
\draw (402.4,149.59) node [anchor=north west][inner sep=0.75pt]   [align=left] {$R_5$};
% Text Node
\draw (274,232) node [anchor=north west][inner sep=-7.75pt]   [align=left] {$\displaystyle \mathcal{L}$};
% Text Node
\draw (26,78) node [anchor=north west][inner sep=-5.95pt]   [align=left] {$\displaystyle \mathcal{L}$};
% Text Node
\draw (496,69) node [anchor=north west][inner sep=0.75pt]   [align=left] {$\displaystyle \mathcal{L}$};
% Text Node
\draw (352,70) node [anchor=north west][inner sep=0.75pt]   [align=left] {$\displaystyle \mathcal{L}$};
\end{tikzpicture}
\end{example}
\section{Simulations}\label{sec:simulations}
We now validate our theoretical findings on synthetic and a standard IEEE benchmark power networks. We examine the dynamical behavior (state evolution) of pseudoinverse Laplacian flows for various network assumptions. It is worth noting that the magnitude of the smallest eigenvalues of the Laplacian and pseudoinverse Laplacian matrix determine the rate of convergence in the corresponding flows respectively. 
Our first set of experiments concerns the unsigned network in Ex.\ref{ex:unsigned} satisfying the assumptions in Thm.\ref{Thm:consensus},\ref{Thm3:unsigned digraph}. As expected, states of both Laplacian and pseudoinverse Laplacian flows attain consensus (Fig.\ref{fig:consensus}). Since the state is complex-valued, we plot the real and imaginary state trajectories separately. 
\begin{figure}[H]
    \begin{subfigure}{0.2\textwidth}
        \includegraphics[scale=0.365]{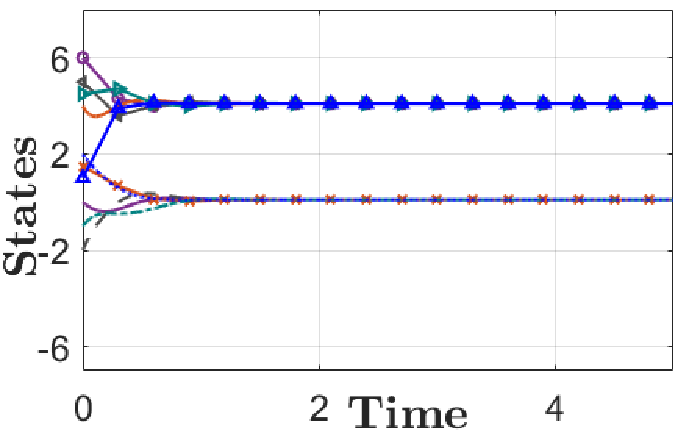}
        \caption{Laplacian flow}
        \label{fig:L_sim wt.balanced_digraph}
    \end{subfigure}
    \hspace{0.67cm} % Adjust the space between the two figures
    \begin{subfigure}{0.2\textwidth}
        \includegraphics[scale=0.344]{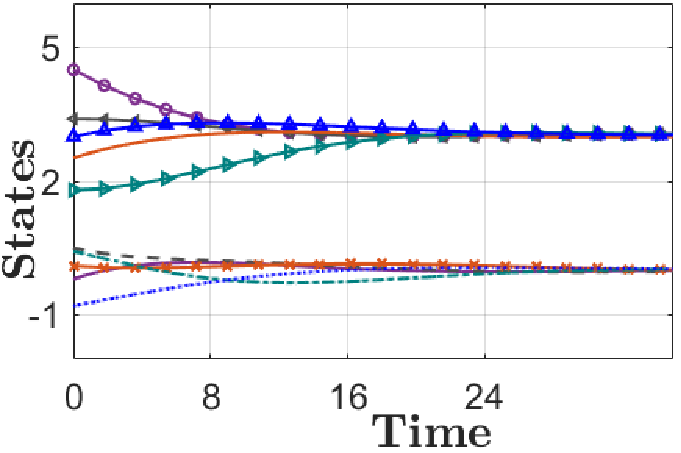}
        \caption{Pseudoinverse flow}
        \label{fig:Linv_sim}
    \end{subfigure}     
    \raggedright
    \caption{\footnotesize Consensus achieved by flow systems for the unsigned digraph in Ex.~\ref{ex:unsigned}. We plot the real and imaginary state trajectories. Observe that the real (upper) and imaginary (lower) states converge to different values.} 
    \label{fig:consensus}
\end{figure}
\begin{figure}[H]
    \begin{subfigure}{0.2\textwidth}
        \centering
        \includegraphics[scale=0.49]{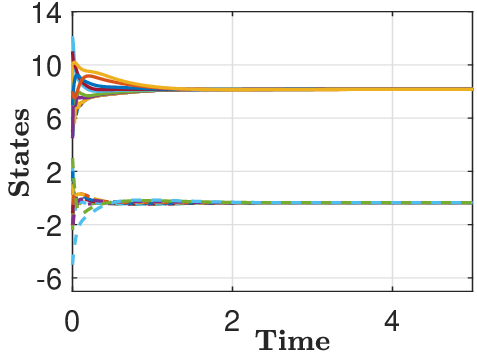}
        \caption{Laplacian flow}
        \label{fig:L_sim signed_graph}
        \end{subfigure}
    \hspace{0.5cm} % Adjust the space between the two figures
    \begin{subfigure}{0.2\textwidth}
    \centering
\includegraphics[scale=0.5]{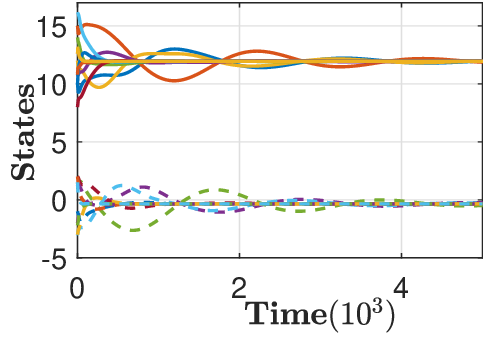}
        \caption{Pseudoinverse flow}
        \label{fig:Linv_sim signed}
          \end{subfigure}     
    \raggedright
    \caption{\footnotesize Consensus achieved by flow systems for the IEEE 10 node network in \cite{zimmerman2010matpower}. Solid and dashed lines denote real and imaginary trajectories, respectively. The notation $10^3$ in (b) represents the time-scale.} 
 \label{fig:L_sim signed}
\end{figure}
\begin{figure}[H]
    \begin{subfigure}{0.2\textwidth}
        \centering
        \includegraphics[scale=0.19]{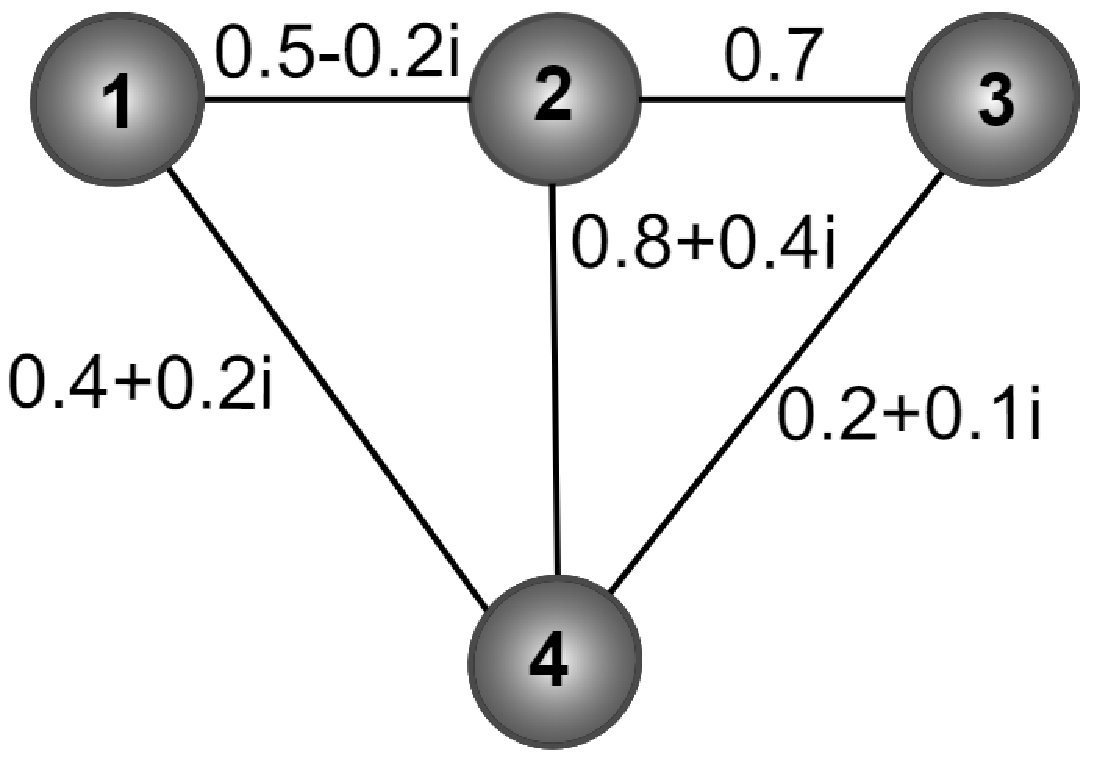}
        \caption{modified impedance network in Fig.\ref{fig:z n/w}}
        \label{fig:RLgraph}
        \end{subfigure}
    \hspace{0.35cm} % Adjust the space between the two figures
    \begin{subfigure}{0.2\textwidth}
    \centering
    \includegraphics[scale=0.52]{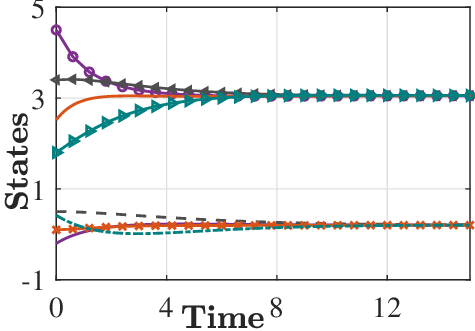}
        \caption{Pseudoinverse flow }
        \label{fig:RL_sim}
          \end{subfigure}     
          \raggedright
    \caption{\footnotesize Consensus achieved by the psuedoinverse Laplacian flow system indicating equal values of current after a transient period in a power network. The exact values of $e^{-L^\dagger t}$ are omitted for brevity.}
    \label{fig:RLsim}
\end{figure}
 Second, we consider the benchmark IEEE power network described in Ex.\ref{ex:signed} and an impedance network in Ex.\ref{ex:RLexample}. We already show that the Laplacian and its pseudoinverse follow rEEP property in Fig.\ref{figs:case10 realEEP}. This implies consensus in both the flows shown in Fig.\ref{fig:L_sim signed} which validates our result in the given real-world power network (as proven in Thm.\ref{Thm:consensus}). Further, Fig.\ref{fig:RLgraph} presents the graph of the impedance network in Fig.\ref{fig:z n/w}. The corresponding pseudoinverse Laplacian flow is simulated in Fig.\ref{fig:RL_sim} which converges to consensus.
\section{CONCLUSIONS AND FUTURE WORK }
We introduce the pseudoinverse Laplacian flows and extend the concept of \emph{real eventually exponentially positivity} (rEEP) to the pseudoinverses of Laplacians in complex-valued networks. Surprisingly, for both undirected and weight-balanced directed graphs, we show that a Laplacian satisfies rEEP if and only if its pseudoinverse does. As a consequence of the rEEP property, both Laplacian and the pseudoinverse Laplacian flows achieve consensus. 
%If we relax the assumption of weight-balance in digraphs, then the study of pseudoinverse Laplacian flows needs exploration.
However, the converge properties of the Laplacian and the pseudoinverse Laplacian flows need further investigation outside of the class of non weight-balanced digraphs. In future work, we plan to investigate whether these assertions apply to all signed digraphs.
       
\bibliography{References}
\bibliographystyle{unsrt}
\end{document}